\documentclass[sigconf,natbib=false]{acmart}

\PassOptionsToPackage{usenames,dvipsnames,table}{xcolor}
\PassOptionsToPackage{
  linkcolor=purple,
  citecolor=purple,
  urlcolor=purple,
  pdfauthor={},
  pdftitle={},
  pdfsubject={},
  pdfkeywords={},
  bookmarks=false}{hyperref}
\usepackage{graphicx} 
\usepackage{booktabs} 
\usepackage{svg} 
\usepackage{xspace} 
\usepackage{algorithm} 
\usepackage{algpseudocode} 
\usepackage{mathtools} 
\usepackage{subcaption} 
\usepackage{amsthm}
\theoremstyle{definition}

\usepackage[
  sorting=none,
  style=numeric-comp,
  backend=biber,
  isbn=false,
  url=true,
  doi=false,
  url=false,
  mincrossrefs=100,
  maxnames=12,
  firstinits=true
]{biblatex}

\AtEveryBibitem{\clearname{editor}}

\newcommand*{\uwfq}{UWFQ\@\xspace}
\newcommand*{\ujf}{UJF\@\xspace}

 \usepackage{array,multirow,graphicx}
 \usepackage{float}

\newcommand{\comalgo}[1]{{\small \textit{\textcolor{blue}{#1}}}}

\newif\ifnotes

\newcommand{\vcutS}{\vspace*{-0.15cm}}

\newcommand{\vcutL}{\vspace*{-0.5cm}}

\title{
Balancing Fairness and Performance in Multi-User Spark Workloads with Dynamic Scheduling 
}

\date{}

\begin{document}

\author{Dāvis Kažemaks}
\email{davis.kazemaks@gmail.com}
\orcid{}
\affiliation{%
  \institution{Delft University of Technology}
  \state{}
  \country{}
}

\author{Laurens Versluis}
\email{laurens.versluis@asml.com}
\orcid{}
\affiliation{%
  \institution{ASML}
  \state{}
  \country{}
}

\author{Burcu Kulahcioglu Ozkan}
\email{b.ozkan@tudelft.nl}
\orcid{}
\affiliation{%
  \institution{Delft University of Technology}
  \state{}
  \country{}
}

\author{Jérémie Decouchant}
\email{j.decouchant@tudelft.nl}
\orcid{}
\affiliation{%
  \institution{Delft University of Technology}
  \state{}
  \country{}
}


\begin{abstract}
Apache Spark is a widely adopted framework for large-scale data processing. 
However, in industrial analytics environments, Spark's built-in schedulers, such as FIFO and fair scheduling, struggle to maintain both user-level fairness and low mean response time, particularly in long-running shared applications. Existing solutions typically focus on job-level fairness which unintentionally favors users who submit more jobs. Although Spark offers a built-in fair scheduler, it lacks adaptability to dynamic user workloads and may degrade overall job performance.  
We present the User Weighted Fair Queuing (\uwfq) scheduler, designed to minimize job response times while ensuring equitable resource distribution across users and their respective jobs. \uwfq simulates a virtual fair queuing system and schedules jobs based on their estimated finish times under a bounded fairness model.   
To further address task skew and reduce priority inversions, which are common in Spark workloads, we introduce runtime partitioning, a method that dynamically refines task granularity based on expected runtime. We implement \uwfq within the Spark framework and evaluate its performance using multi-user synthetic workloads and Google cluster traces. We show that \uwfq reduces the average response time of small jobs by up to 74\% compared to existing built-in Spark schedulers and to state-of-the-art fair scheduling algorithms. 
\end{abstract}


\maketitle

\section{Introduction}

The exponential growth of online data~\cite{apoorv_background_data_increase} has driven an increasing demand for scalable, high performance batch processing frameworks. Among these, Apache Spark, or simply Spark, has emerged as a leading platform~\cite{tang_background_spark_usage} due to its in-memory computation model, rich ecosystem of libraries (e.g., SQL engine~\cite{armbrust_background_spark_sql}) and integration with major cloud providers such as Databricks and Amazon Web Services.  
In industrial analytics environments, Spark is commonly deployed as a shared, multi-user system, where users concurrently submit batch jobs for execution. 
These environments rely on fair scheduling mechanisms to manage constrained computing resources and deliver acceptable job response times to all users. 

In most systems, such a fair scheduler is typically implemented through cluster managers such as YARN or Apache Mesos to control the job flow before it reaches Spark's scheduler. Users insert their main execution code within a Spark application, which additionally interfaces with the cluster manager to request resources. Cluster managers assign resources across multiple Spark applications, each launched by a different user. Prior works have explored such approaches, focusing on equalizing resource distribution~\cite{wei_chen_scheduling_spark_preemption} or meeting job deadlines through application-level scheduling~\cite{wang_scheduling_application_level_spark_yarn}. 
Additionally, Spark now offers dynamic resource allocation (DRA), which adjusts the amount of resources based on the application workload. However, DRA does not enforce strict fairness policies nor explicitly aim to minimize execution time, in particular when the system is congested.

However, launching an application for each workload induces a significant performance overhead in environments where many small jobs need to be executed. 
This overhead comes from to the time required to start the driver program, allocate workers, or launch executors~\cite{chen_background_spark_scheduling_delay}. Furthermore, if the expected workloads are highly parallelizable, all applications benefit from having more executors.

Industrial analytics platforms therefore typically rely on a single long-running Spark application that continuously listens for incoming user jobs and executes them for performance. 
While this solution reduces setup delays, it centralizes scheduling responsibilities within Spark's scheduler, which is ill-equipped to handle complex fairness and performance trade-offs. Spark's built-in scheduling options, i.e., first-in-first-out (FIFO), fair, and fairness pools, fall short in achieving both fairness across users and low job response times, particularly in multi-user long-running Spark applications where job arrival patterns and workloads are highly dynamic~\cite{chen_li_scheduling_mix_min_geo_centers,chen_scheduling_fair_cluser_spark}.
Alternative schedulers have been proposed to ensure fairness among jobs~\cite{li_scheduling_clustering_fair_metaheuristic, quang_scheduling_dynamic_fairness, chen_scheduling_fair_cluser_spark, ferreira_scheduling_fairness_nonclairvoyant, ilyushkin_scheduling_impact_of_unknown_runtime,jia_scheduling_metaheuristic_fairness, pastorelli_scheduling_virtual_sizes}. However, by focusing on job fairness, these schedulers allow users who schedule more jobs to be allocated more resources than those with fewer jobs, which is unfit for our multi-user and multi-job batch processing environment.

Task skew and priority inversions accentuate fairness issues in Spark environments. Both problems stem from Spark's default task partitioning strategy, which decomposes a job's input data into multiple partitions, each executed in parallel as an individual task. 
Task skew occurs when one or more partitions take significantly longer to execute than others. This leads to underutilization of resources, as the straggling partition may use a single core while other cores remain idle, effectively reducing the job's parallelism and extending its completion time. 
Priority inversions occur when long-running tasks from lower-priority jobs occupy executors, thereby blocking higher-priority jobs from making progress. Such inversions 
have been observed in Spark deployments~\cite{chen_chen_scheduling_speculative_reservation} and are particularly problematic when poor partitioning results in jobs that cannot be preempted or rescheduled efficiently~\cite{chen_scheduling_fair_cluser_spark}.

We propose User Weighted Fair Queuing (\uwfq), a novel Spark scheduler that ensures bounded user-job fairness  while significantly reducing job response times. 
\uwfq builds upon the virtual time concept introduced by Cluster Fair Queuing (CFQ)~\cite{chen_scheduling_fair_cluser_spark} by incorporating both user and job context awareness. Additionally, \uwfq dynamically adjusts task granularity based on estimated runtimes to minimize task skew and mitigate priority inversions.  

As a summary, we make the following \textbf{contributions}:

    We introduce \uwfq, a new scheduler that ensures bounded user-job fairness while reducing the average response time of jobs. \uwfq operates by simulating a virtual fair scheduler on currently active jobs in the system, and assigning the highest priority to jobs that would complete the earliest in the virtual scheduler.
    
    We propose a dynamic partitioning method that utilizes the estimated runtime of jobs to determine the optimal number of input partitions to use, thereby avoiding priority inversion caused by limited preemption in Spark and preventing job response time extension due to task runtime skew. This dynamic partitioning algorithm provides more flexibility to \uwfq and further improves the system's performance. 
    
    We implement UWFQ with dynamic partitioning in the Spark framework, and evaluate its effects on fairness and job average response times. We consider synthetic micro-user workloads that highlight certain scenarios that most schedulers handle inefficiently, and real Google cluster traces~\cite{google_traces_wta_format_2014}.
    Our results indicate that \uwfq with runtime partitioning reduces the response times of small and medium-sized jobs by up to 74\% in homogeneous workloads compared with a plain user-job fairness implementation. We additionally showcase that UWFQ always does equally good or better than CFQ, and almost always outperforms Spark's built-in fair scheduler.

\section{System Model and Objectives}

This section provides necessary background on Spark and specifies our  objectives. 

\subsection{Apache Spark}

Apache Spark is an analytics engine for large-scale data processing across multiple programming languages. It includes a rich set of modules that support diverse workloads, including SQL querying, real-time stream processing, machine learning and graph analytics. Most of these high-level operations are compiled down to low-level API calls on Resilient Distributed Datasets (RDD), Spark's core abstraction for distributed data processing. 

\subsubsection{Job execution}

\autoref{fig: job exec} illustrates a simplified job execution. Users construct their jobs by applying transformations on RDDs and finalizing results by issuing an action, which internally triggers a Spark job (Step 1). Once a Spark job is submitted, it is picked up by the DAG Scheduler, which breaks down the job into stages and constructs a DAG dependency graph between them. This graph is created by first creating a result stage from the associated action and recursively adding dependent map stages until no dependencies are found. Stage input is then partitioned into tasks (Step 2) and submitted to the Task Scheduler if all of their dependencies are satisfied (Step 3). The Task Scheduler keeps track of all schedulable units using a Root Pool that contains both individual stages and other layers of pools used for establishing priority hierarchies. 

\begin{figure}[t]
    \centering
    \includesvg[width=\columnwidth]{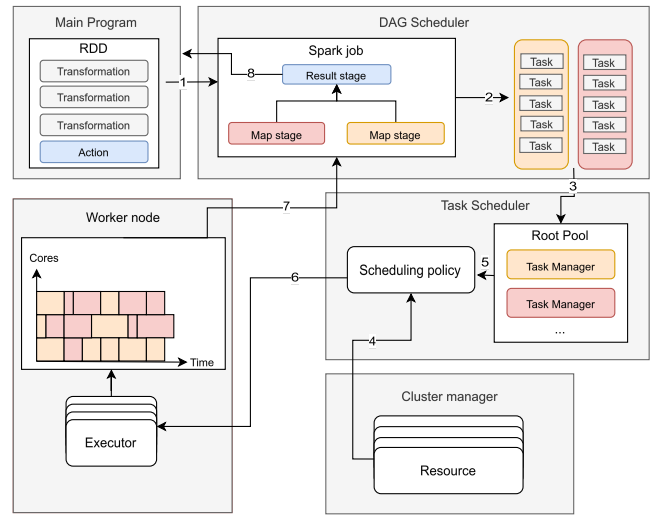}
    \vcutS
    \caption{Execution of user-constructed RDDs on cluster resources.}
    \label{fig: job exec}
    \vcutL
\end{figure}

\begin{figure*}[t]
    \centering
    \includesvg[width=0.9\textwidth]{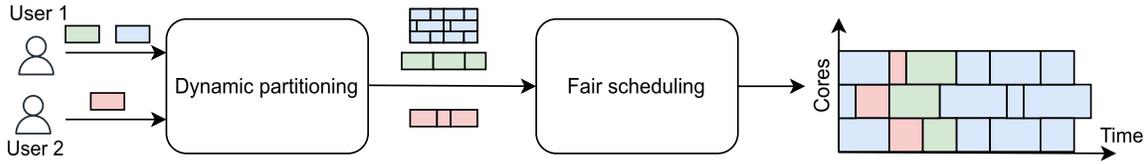}
    \caption{System outline. Here, user 1 schedules a long-running job (blue), followed by a short job (red) from user 2 and another smaller job by user 1 (green). All users and jobs benefit from a low response time thanks to dynamic partitioning and fair scheduling.}
    \label{fig: system outline}
\end{figure*}

Within the Task Scheduler, stages are represented as Task Managers that additionally keep track of their task state. Once the cluster manager offers resources to the Task Scheduler (Step 4), it sorts the Root Pool (Step 5) according to the defined scheduling policy and submits each task one by one to be run on the executors (Step 6). Once all tasks of a Task Manager have been executed, the stage dependency graph is updated (Step 7), and any stage without pending dependencies is submitted to the Task Scheduler. Once the result stage is finished, the result is returned to the Main Program (Step 8). While scheduling decisions are mainly controlled by a scheduling policy, the scheduler always adheres to the dependency structure imposed by other layers. 

\subsubsection{Partitioning}

Stage input data is split into multiple partitions that can be operated on in parallel. For each partition in a stage, a task is created, which performs the stage-defined operations on its associated data partition. This is the main concept that allows Spark to parallelize job execution.
Partitioning of stage data can be performed in two distinct phases: once when the data is first loaded into a stage, and during Adaptive Query Execution (AQE) optimizations, where partitions can be coalesced to reduce their number. When data is first retrieved, the number of partitions is determined by dividing the data equally among the available cores on the allocated executors. This allows Spark to maximize the parallelization of stage execution on executors. AQE, on the other hand, starts with 200 partitions and then reduces them to a more appropriate number based on the recommended partition size, or again, trying to maximize parallelism.

\subsubsection{Built-in Schedulers}

Spark has two built-in scheduling algorithms: first-in first-out (FIFO) and fair. The FIFO policy schedules Spark jobs in the order of their arrival, while the fair scheduler tries to equalize the number of running tasks across all submitted stages. To introduce more fairness levels, Spark also provides fairness pools that allow for defining a fairness hierarchy among jobs. During task scheduling, first, the pool with the highest priority is selected using the root scheduling policy. Then, within this pool, the task with the highest priority is selected according to a pool scheduling policy. The two pool scheduling policies that are currently offered are FIFO and Fair. Once the priority of a task is determined, it is allocated to the resource with the highest task locality. 

While Spark provides several scheduling options, it does not support fair user-level scheduling for long-running applications. Its fair scheduling algorithm only considers job-level fairness, rather than user-level fairness, creating conditions where users with more active stages receive more resources. While fairness pools provide a high degree of configurability to create fairness among users, the user pools are configured only at the start of the application. This means that for a long-running application with a dynamic user base, this option is also insufficient.

\subsection{Objective: User-Job Fairness}

Our target analytics platform accommodates multiple users who may schedule several concurrent jobs. Each user in the system is entitled to an equal share of the system resources and would like their jobs to execute as quickly as possible. Jobs only provide utility once they are fully finished; hence, having a steady progression of jobs is not necessary. However, jobs should not be starved and finish within a reasonable amount of time relative to their runtime. Jobs in the system have no explicit priority, hence, job priority should only be derived from their submission time and runtime.

Since we do not find any existing concrete definition of fairness that would perfectly fit the target system's needs, we extend an existing definition, namely, fair share scheduling, to encapsulate our fairness objective.
We define \textbf{User-Job Fairness (\ujf)} as an extension of fair share scheduling, where weights for jobs would be obtained by distributing resource shares evenly among active users, and then distributing user shares among jobs belonging to that user. We formulate this using 
$R_{k} = \frac{R}{N_u}$ and $R_i = \frac{R_k}{N^k_j}$,
where $R$ is the total amount of resources, $R_{k}$ is user $k$'s fair share, $N_u$ is the number of active users, $R_i$ is the share for job $i$, and $N^k_j$ are active jobs of user $k$. User-job fairness is achieved when resources are distributed exactly proportionally to all jobs' shares.
This definition provides multi-level fairness, where each user is always entitled to an equal amount of resources, and jobs associated with the same user do not steal resources from other users, while still ensuring gradual progression without starvation. 
We also define bounded \ujf, where a scheduling algorithm is considered bounded by \ujf if every job's finish time $f_i$ is within a constant time $C$ of the time it would finish within a UJF schedule $\hat{f_i}$, i.e., $f_i - \hat{f_i} \leq C $.

\section{\uwfq: User Weighted Fair Queuing}

In this section, we present \uwfq, a Spark scheduler that implements response time-efficient user-job fair scheduling, and a dynamic partitioning algorithm that uses job runtimes to improve parallelizability of jobs, as illustrated in ~\autoref{fig: system outline}.
Appendix~\ref{sec:proof} provides a proof of correctness for \uwfq.
When a user schedules a job, its input data is first partitioned into smaller slices of data using a dynamic partitioning algorithm to reduce the skews and long executor reservation times caused by large tasks, which can negatively impact performance. 
Once the job's input is partitioned, it is forwarded to the fair scheduler to determine its priority. \uwfq determines the priority of each job by estimating the finish time of each job in a user-job fair scheduling system and assigning priority to jobs in the order of their completion time. Finally, tasks of each job are scheduled onto the executors in the order of their priority.

\subsection{Job Context Awareness}

In Spark, the highest level of work unit abstraction is a Spark job, which encapsulates all job priority-related metadata. However, this is insufficient for an analytics system, since a single analytics job may represent multiple Spark jobs. Because users only receive utility from analytics jobs that have fully finished, it is more efficient to model job priorities with respect to their highest abstraction level.
We therefore embed this job context into every Spark job that is created, such that it will be scheduled with respect to its corresponding workflow or query. This allows jobs to be scheduled in a much more effective sequence to optimize the response times, rather than having to interleave with every simultaneous job.

\subsection{Dynamic Partitioning}

Spark uses the estimated size of input data to partition stages across multiple executors. This partitioning method ensures that each stage can run on all cores in the system (if available), maximizing parallelism and throughput of the underlying job. However, this partitioning does not account for the runtime of these partitions and can introduce skews and priority inversion in certain cases.

For example, one of the partitions in~\autoref{fig: partitioning skew}~(\subref{subfig: skew bad}) runs 5x longer than the others. Because of this skew, the system does not fully utilize all its available resources, and the response time of the job is delayed. 
Priority inversion can also be observed, e.g., when a long job completes before a higher priority job can receive any resources to execute on, as illustrated in~\autoref{fig: partition prioirty inversion}~(\subref{subfig: inversion bad}). Since Spark tasks are not preemptable, if a longer job arrives just before a higher priority job, it cannot be interrupted and will delay the higher priority job.

\begin{figure}[t]
    \centering
    \begin{subfigure}{0.4\textwidth}
        \centering
        \includesvg[width=\linewidth]{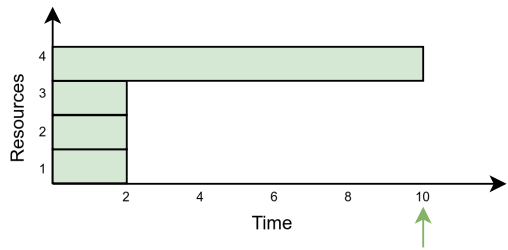}
        \caption{Stage data partitioned using static partitioning}
        \label{subfig: skew bad}
    \end{subfigure}
    \hfill
    \begin{subfigure}{0.4\textwidth}
        \centering
        \includesvg[width=\linewidth]{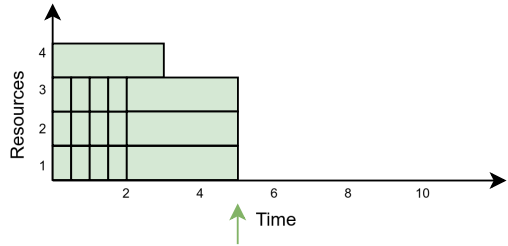}
        \caption{Stage data partitioned using \uwfq's partitioning}
        \label{subfig: skew good}
    \end{subfigure}

    \caption{Impact of task skew on job runtime. The default task partitioning leads to high completion time, in~(\subref{subfig: skew bad}), while \uwfq's dynamic partitioning, in~(\subref{subfig: skew good}), reduces it. The green arrows indicate the finish time of the job. 
    }
    \label{fig: partitioning skew}
\end{figure}

We introduce runtime partitioning to mitigate skews and priority inversions caused by default Spark partitioning of the stage input data. Runtime partitioning attempts to split the stage input into partitions that run in constant time, allowing tasks to be distributed more evenly and reducing the maximum amount of time a task can reserve an executor. The number of partitions is computed as $\frac{Stage\ runtime}{ATR}$ while their size is computed as $\frac{Total\ input\ size}{Partition\ amount}$,
where Advisory Task Runtime ($ATR$)  is the user or system-defined desired task runtime that the partitions are expected to run for. First, we need to collect or estimate accurate $Stage\ runtime$s, which are necessary to perform this partitioning method. Once the stage arrives at the dynamic partitioning component, $Partition\ amount$ can be calculated to determine the number of partitions the data will be split into, which is then converted into $Partition\ size$ that will be practically used to distribute the input data across partitions. After splitting the data, a task will be created for each partition, which will then be eventually scheduled on the executor according to its job's priority.

By adjusting the $ATR$ value, users or the system can control the granularity of all tasks running on the executors. By setting a relatively low advisory task runtime, we can mitigate both skews and priority inversions that were present in the original partitioning method. Skews are avoided by applying more aggressive partitioning, where the stage data can be split into significantly more partitions than there are cores in the system. Because of this, the skewed partitions are additionally split, which allows for spreading the runtime more evenly across executors, as seen in~\autoref{fig: partitioning skew}~(\subref{subfig: skew good}). Priority inversions are also mitigated by this, as tasks release executors much faster, allowing higher-priority tasks to be assigned much quicker than they would under regular partitioning. However, some overhead is introduced by having many active tasks in the system, hence, the $ATR$ value should not be set too low.

\begin{figure}[t]
    \centering
    \begin{subfigure}{0.4\textwidth}
        \centering
        \includesvg[width=\linewidth]{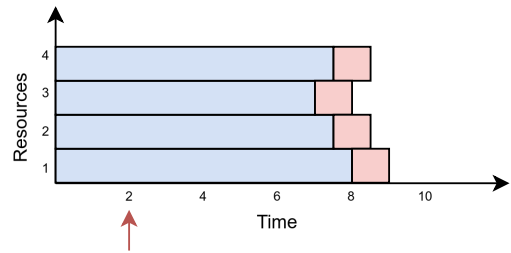}
        \caption{Priority inversion using static partitioning}
        \label{subfig: inversion bad}
    \end{subfigure}
    \hfill
    \begin{subfigure}{0.4\textwidth}
        \centering
        \includesvg[width=\linewidth]{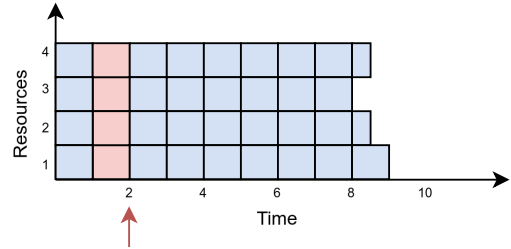}
        \caption{Priority inversion using \uwfq's partitioning}
        \label{subfig: inversion good}
    \end{subfigure}

    \caption{Priority inversion of red (high priority) and blue (low priority) jobs. In~(\subref{subfig: inversion bad}) the red job has a smaller runtime and higher priority, but only executes once the blue job's tasks have finished. But in~(\subref{subfig: inversion good}), thanks to \uwfq's dynamic partitioning, the red job is able to execute earlier. The red arrow indicates the expected finish time of the red job.}
    \label{fig: partition prioirty inversion}
\end{figure}

\subsection{Fair Scheduling}

\uwfq reuses principles originally proposed by Parekh et al.~\cite{parekh_scheduling_original_weighted_fair_queuing} to design WFQ. Specifically, it uses virtual time to efficiently compute job finish times under Generalized processor sharing (GPS) and leverages these virtual finish times to schedule jobs that can complete earlier than they would under GPS, while remaining bounded by it. However, these principles must be adapted to work for a user-job fair environment.
We use virtual time in 2 layers: user virtual time and global virtual time. User virtual time is assigned to each user and functions similarly to regular virtual time, serving as the basis for obtaining job completion times under GPS for that specific user. Global virtual time is shared across all users and is used to determine global virtual deadlines of jobs, which are used to establish priority among all active jobs in the system.

For example, in a scenario where users compete for the same static resources, each having highly parallelizable workloads. By equally allocating resources to each user, the runtimes of jobs are extended due to not achieving optimal parallelism. However, if we were to run these jobs sequentially, each job would be able to utilize all available resources. By assigning a deadline to each job based on its expected finish time, we obtain a global order that allows quicker jobs to complete faster, while still ensuring that large jobs eventually complete once their deadline approaches. To calculate these deadlines, each user first orders its jobs using user virtual time, and then assigns a deadline to each job with respect to global virtual time, such that jobs of multiple users can be compared.

\subsubsection{Algorithm}

We show the assignment of job deadlines in~\autoref{algo: uwfq insert}. To create a schedule, we order the jobs based on their assigned deadlines. 
The algorithm centers around the user entity $U_k$ that keeps track of its associated jobs and virtual time $V^k_{user}$. User virtual time $V^k_{user}$ is used to calculate user deadlines $D_{user}$ for incoming jobs $J_i$. Job order depends on both the user virtual time $V^k_{user}$ upon arrival and the job's slot-time $L_i$. We define job slot-time $L_i$ as the time needed to execute all of job's tasks on a single core sequentially. To allow for dynamic priorities, we also add a scalar $U_w$ for each user. This is set to 1 in the case where users have equal priorities, but can be adjusted to favor certain users in the system, 
Once the job is ordered in the user job set $S^k_{jobs}$, its global deadline $D^i_{global}$ can be derived. This is done by utilizing job's slot-time $L_i$ and user virtual arrival time $V^k_{arrival}$, which is measured from the current virtual global time $V_{global}$ once the user schedules their first job. The deadline is then calculated by sequentially adding each job's scaled $L_i$ to $V^k_{arrival}$ to obtain the expected virtual finish times of jobs across all users.

{
\renewcommand{\baselinestretch}{1.3}
\begin{algorithm}[t]
\caption{Job deadline assignment under UWFQ}
\label{algo: uwfq insert}
\begin{algorithmic}[1]

{\small

\Statex \textbf{Globals:} A set $S_{users}$ of tuples of users $U$ and their virtual arrival times $V_{arrival}$ ;\space Global virtual time $V_{global}$

\Statex \textbf{Input:} Current time $T_{current}$;\space User $k$ $U_k$;\space  Job $i$ arrival time $T^i_{arrival}$;\space  Job $i$ $J_i$ duration $L_i$;  

\Statex \comalgo{// Phase 1: update system}
\State $\Call{updateVirtualTime}{T_{current}}$ \Comment{See \autoref{algo: update virtual time}}

\If{$(U_k,\space \_) \notin S_{users}$ }
\State $V^k_{arrival} \gets V_{global}$
\State $S_{users} \gets S_{users}\cup (U_k, V^k_{arrival})$ 
\EndIf

\Statex \comalgo{// Phase 2: calculate user deadline and insert job $J_j$ into set of user jobs}
\State $V^k_{user} \gets \Call{getUserVirtualTime}{U_k}$ 
\State $D^i_{user} \gets V^k_{user} + L_i * U_w$ 
\State $S^{k}_{jobs} \gets \Call{getUserJobs}{U_k} $ \Comment{$S_{jobs}$ is a sorted set of tuples on virtual user deadlines $D_{user}$}
\State $S^{k}_{jobs} \gets S^{k}_{jobs} \cup (J_i,\space L_i,\space  D^i_{user})$

\Statex \comalgo{// Phase 3: update user job global virtual deadlines}
\State $(J_{first},\space L_{first}) \gets \Call{getEarliestUserDeadline}{S^{k}_{jobs}}$
\State $V^k_{arrival} \gets \Call{getUserVirtualArrivalTime}{U_i}$
\State $D^{first}_{global} \gets V^k_{arrival} + L_{first} * U_w$
\State $\Call{setJobDeadline}{J_{first}, D^{first}_{global}}$

\State $D^{previous}_{global} \gets D^{first}_{global}$
\For{each $J_a$ in $S^{k}_{jobs}$ sorted by $D^a_{user}$ and $J_a \not= J_{first}$}
    \State $D^a_{global} \gets D^{previous}_{global} + L_a * U_w$
    \State $\Call{setJobDeadline}{J_{a}, D^{a}_{global}}$
    \State $D^{previous}_{global} \gets D^a_{global}$

\EndFor
}
\end{algorithmic}
\end{algorithm}
}

\subsubsection{Updating virtual time}
\label{sec: updating virtual time}

In 2-level virtual time, we have to update both the global time and virtual time for each user. \autoref{algo: update virtual time} updates the global virtual time and \autoref{algo: update virtual time users} updates a user's virtual time. We split these algorithms in functions that we explain in the following. 

{
\renewcommand{\baselinestretch}{1.3}
\begin{algorithm}[t]
\caption{Virtual time updating}\label{algo: update virtual time}
\begin{algorithmic}[1]
{\small

\Statex \textbf{Globals:} A set $S_{users}$ of tuples of users $U$ and their virtual arrival times $V_{arrival}$ ;\space Global virtual time $V_{global}$;\space Previous update time $T_{previous}$;\space Total resources of the system $R$

\Function{updateVirtualTime}{$T_{current}$}
\For{each $(U_k,\space \_)$ in $S_{users}$ sorted by $D_{global}$}
\State $R_{user} \gets \frac{R}{|S_{users}|}$
\State $T^k_{finish} \gets$ \Call{getUserFinishTime}{$U_k, R_{user}$}
\If{$T^k_{finish} > T_{current}$}
    \State \textbf{break}
\EndIf
\State $S_{users} \gets S_{users} \setminus (U_k,\space \_)$
\State \Call{progressVirtualTime}{$T^k_{finish}, R_{user}$}

\EndFor

\State $R_{user} \gets \frac{R}{|S_{users}|}$
\State \Call{progressVirtualTime}{$T_{current}, R_{user}$}

\EndFunction

\Function{getUserFinishTime}{$U, R_{user}$}
    \State $D^{latest}_{global} \gets \Call{getLatestDeadline}{U}$
    \State $T_{spent} \gets (D^{latest}_{global} - V_{global}) / R_{user}$
    \State $T_{finish} \gets T_{previous} + T_{spent}$
    \State \textbf{return} $T_{finish}$
\EndFunction

\Function{progressVirtualTime}{$T, R_{user}$}
    \State $T_{passed} \gets T - T_{previous}$
    \State $V_{global} \gets V_{global} + T_{passed} * R_{user}$

    \For{each $(U_k,\space \_)$ in $S_{users}$}
        \State \Call{updateUserVirtualTime}{$U_k, R_{user}, T$}
    \EndFor

    \State $T_{previous} \gets T$
\EndFunction
}
\end{algorithmic}
\end{algorithm}
}

{
\renewcommand{\baselinestretch}{1.3}
\begin{algorithm}[t]
    \caption{Virtual time updating for users}
    \label{algo: update virtual time users}
    \begin{algorithmic}[1]
    {\small
    \Statex \textbf{Globals:} \space Previous update time $T_{previous}$
    
    \Function{updateUserVirtualTime}{$U_k, R_{user}, T_{current}$}
        \State $S^{k}_{jobs} \gets \Call{getUserJobs}{U_k}$
        \State $T^{user}_{previous} \gets T_{previous}$
        \State $V^k_{user} \gets \Call{getUserVirtualTime}{U_k}$ 
        \For{each $(J_i,\space L_i,\space  D^i_{user})$ in $S^{k}_{jobs}$ sorted by $D_{user}$}
    
            \State $R_{job} \gets \frac{R_{user}}{|S^{k}_{jobs}|}$
            \State $T_{passed} \gets T_{current} - T^{user}_{previous}$
            \State $V^i_{user} \gets V^k_{user} + (T_{passed} * R_{job})$
            \If{$D^i_{user} > V^i_{user}$}
                \State \textbf{break}
            \EndIf
            \State $V_{spent} \gets D^i_{user} - V^k_{user}$
            \State $T_{spent} \gets \frac{V_{spent}}{R_{job}}$
            \State $V^k_{user} \gets V^k_{user} + V_{spent}$
            \State $T^{user}_{previous} \gets T^{user}_{previous} + T_{spent}$
            \State $V^k_{arrival} \gets \Call{getUserVirtualArrivalTime}{U_k}$
            \State $V^k_{arrival} \gets V^k_{arrival} + L_i$

            \State $S^{k}_{jobs} \gets S^{k}_{jobs} \setminus (J_i,\space L_i,\space  D^i_{user})$
        \EndFor
        \If{$|S^{k}_{jobs}| > 0$}
            \State $R_{job} \gets \frac{R_{user}}{|S^{k}_{jobs}|}$
            \State $T_{spent} \gets T_{current} - T^{user}_{previous}$
            \State $V^k_{user} \gets V^k_{user} + (T_{spent} * R_{job})$
        \EndIf

    \EndFunction

    }

\end{algorithmic}
\end{algorithm}
}

\textbf{updateVirtualTime}.
To update virtual time, we must iterate over all existing users in the order of their completion time. This is necessary to advance the global virtual time correctly, as it progresses at a rate proportional to the number of active users. We first determine the user's share, $R_{user}$, and use it to compute the user's actual finish time, $T^k_{finish}$. If $T^k_{finish}$ is after the current time $T_{current}$, we can conclude that the current and following users have not finished their jobs yet. If $T^k_{finish}$ is before current time $T_{current}$, then each remaining user should update until $T^k_{finish}$ with the current user share $R_{user}$, to have their virtual times up to date before the user leaves the system and distributes their share among active users.
Once we have handled all leaving users, we can progress virtual time till the current time $T_{current}$. Since no user will leave the system during this period, the user share, $R_{user}$, will remain consistent and progress all users equally.

\textbf{getUserFinishTime}.
We obtain the user's finish time $T_{finish}$ as the time its last job finishes. Since all user jobs are ordered based on their user virtual deadlines, we can trivially get the latest finished job by taking the last element from the user job set. To convert from global virtual deadline $D^{latest}_{global}$ to real time, we can calculate the difference between current global virtual time $V_{global}$ and the deadline $D^{latest}_{global}$, and divide it by the user share $R_{user}$. The difference represents the global virtual time that will progress between now and the time the job ends, and dividing by the user share $R_{user}$ allows us to convert from virtual to real time units.
Finally, we obtain the finish time $T_{finish}$ by adding the real time spent $T_{spent}$ to the previous update time $T_{previous}$. Previous update time represents the period when virtual time was last updated, or the previous current time $T_{current}$ that was used to update virtual time.  

\textbf{progressVirtualTime}.
To progress virtual time, both global and user virtual times must be updated to the current time $T$.
We first progress the global virtual time. This can be done by calculating the amount of real time that has passed $T_{passed}$ since the previous update $T_{previous}$, and then using it to calculate the virtual time that has passed by multiplying it with the user share amount $R_{user}$. Virtual time in this context represents the marginal rate of progress each user experiences.
Then, to update the user virtual time, we iterate over all active users and progress them till the current time $T$. We cover the details of this in the following subsection.
Lastly, we update the previous update time $T_{previous}$ to reflect the time $T$ it was updated to.

\subsubsection{updateUserVirtualTime}
Updating user virtual time requires iterating over all currently active user jobs in the order of their virtual deadlines $D^j_{user}$. This is because of the same principle as for global virtual time, to ensure that jobs progress at the correct marginal rate. 
We first need to set the user's previous update time, $T^{user}_{previous}$, to the previous update time, $T_{previous}$. This time will be used to track the periods between virtual time updates, ensuring the correct progression of virtual time forward. 
Then we iterate over all user jobs in the order of their user virtual deadlines $D_{user}$. We first calculate the current job share, $R_{job}$, and the time that has passed since the previous update, $T_{passed}$. These arguments can then be used to calculate the assumed user virtual time $V^i_{user}$ that does not account for jobs leaving the system, hence does not equate to the actual user virtual time $V^k_{user}$. However, we can use this virtual time to test whether the job with the earliest deadline, $D^i_{user}$, would have finished. 
If the assumed user virtual time $V^i_{user}$ is after the earliest deadline $D^i_{user}$, we can determine that no job will finish before the current time, and break the loop. However, if the job has finished, we then calculate the virtual time $V_{spent}$ spent on this job by taking the difference between its deadline $D^i_{user}$ and the current user's virtual time $V^k_{user}$. This virtual time can then be converted into real time, $T_{spent}$, by dividing it by the job shares, $R_{job}$. We update the user virtual time $V^k_{user}$ by adding the virtual time that has been spent $V_{spent}$, and progress the previous update time $T^{user}_{previous}$ by the real time spent $T_{spent}$. Besides the user's virtual time, we also must update the virtual arrival time $V^k_{arrival}$. This is done by progressing the virtual arrival time $V^k_{arrival}$ by the runtime of the job $L_i$. The virtual arrival time must be progressed, so future jobs that are assigned the global deadlines account for jobs that have finished in the past, and keep global order consistent. 
Finally, once all finishing jobs are accounted for, in the case there are still any jobs left in $S^k_{jobs}$, we must progress the user virtual time $V^k_{user}$ forward until it is caught up to the current time $T_{current}$. We can do this by calculating the remaining time $T_{spent}$ and multiplying it by the job shares $R_{job}$ to obtain the virtual time that has passed, and then adding it to the user's virtual time $V^k_{user}$.

\section{Implementation Details}

This section discusses how \uwfq is integrated into Spark and how it accounts for the dynamic runtime environment.

\subsection{Apache Spark Integration}

The integration within Spark requires addressing three parts: scheduler, partitioner, and external system integration. 

\subsubsection{Scheduling}
Spark already provides an extendable interface for adding new schedulers. However, this interface is only accessible within the confinement of the source code. To make the framework more flexible, we extend the Spark source code with class loading for custom schedulers. This allows the Spark driver to load any precompiled Java Spark scheduler that is passed as an argument. 
We have released our modified Spark 3.5.5 version\footnote{\href{https://github.com/kazemaksOG/spark-3.5.5-custom}{https://github.com/kazemaksOG/spark-3.5.5-custom}}.
In current Spark releases, the scheduler only operates at the level of stages, where the stage with the highest priority is determined, and its tasks are scheduled onto the executors whenever resources are available. \uwfq is implemented in this layer, where incoming stages have their priority assigned by \uwfq. However, instead of isolating the priority for stages, each stage is first mapped to its corresponding analytics job from which the priority is inherited. Based on when the job first arrived in the system and its total runtime across all stages, the virtual deadline is determined, which will be assigned to every stage belonging to this analytics job. This ensures that even if a job is spread across multiple stages, it will still be executed within the guaranteed user-job fair time bounds and avoid unnecessary interleaving. 
In the Spark ecosystem, the highest priority is assigned to the stage with the lowest priority value $P_s$. This aligns perfectly with global virtual deadlines, as the job with the lowest deadline value $D$ should be scheduled the earliest. We express the priority assignment as $P_s = D^i_{global}$.

\subsubsection{Partitioning}

Spark performs partitioning in two distinct phases of job execution: during the initial input reading and when coalescing between shuffle stages.
When a Spark job is submitted, it must perform a file scan to load the necessary data into the first stage. During the file scan, partition sizes are calculated only using the input data size, which is then used to divide the given input files between the tasks. To introduce runtime partitioning, we override this function with custom partitioning, which then attempts to partition the stage input using estimated runtime. 
After finishing the leaf stages of the DAG, there may be multiple shuffle stages that funnel the output of previous stages as input for the upcoming stage. By default, all shuffle stage outputs are written into 200 partitions. However, the Adaptive Query Execution (AQE) module coalesces these partitions into more appropriately sized partitions using the intermediate output sizes, without considering the upcoming stage runtime. Because the minimum partition amount is set to 1 by default in this phase, it can create long-running tasks if the following stage runtimes are not considered. To improve AQE coalescing, we replace the default minimum partition amount with our dynamic amount calculated from the estimated stage runtime.
This ensures that the partitions never coalesce down to an amount that would introduce long-running tasks, and also minimally interferes with AQE's own methods of reducing runtime skews.
To add more flexibility to the framework, in both cases, the partitioning algorithm is class-loaded into the framework on startup, similarly to the scheduler.

\subsubsection{External system}

The external layer that uses Spark's engine for batch processing only has to interact with Spark Context to ensure compatibility with \uwfq. This is because all metadata that is embedded into the Spark Context is attached to the submitted job that will propagate the metadata to subsequent layers, allowing \uwfq to make proper scheduling decisions. 
\uwfq only needs 2 essential properties to be provided when scheduling a user job: user context and job context. User context allows \uwfq to map each of the stages to their corresponding user, to enable user-job fairness in the schedule. Job context is used to group together stages that were submitted under the same analytics job, allowing stages to be scheduled with respect to the job they belong to instead of being independent. Since users are only concerned with the final outcome of their analytics job, the intermediate results of individual stages do not matter.
Additionally, \uwfq needs runtime predictions to perform effective scheduling, which can be provided by the external system or a performance estimation module within Spark's engine. We achieve this by adding a class-loaded performance estimator to the framework, which can then be accessed across all Spark components to inquire performance metrics of a stage or other units. The benefit of this component is decoupling the performance prediction from the external system, and providing a much more specialized interface for other components to acquire runtime estimates of work units.

\subsection{Accounting for Dynamic Environments}

\uwfq stops considering users as soon as all of their jobs have finished. This is necessary to accurately distribute resources among users who still have pending jobs running. However, due to delays caused by inaccuracies in runtime estimation, there are cases where users would exit the system before all the stages associated with their job have finished. This creates a scenario where a discrepancy exists between the real system and virtual time, where the real job is still executing even though it has finished in the virtual scheduler. If the user's job still has stages that have not been scheduled, these stages will not be attributed to the correct start time.
To correct this, we introduce a grace period, in which the following stages can “revive” a user who has exited the system with their original arrival time. 

While this provides the solution for these scenarios, it presents some caveats. By setting the grace period too small, it will not be triggered in the highlighted scenarios, rendering it useless. However, setting it too high could allow inaccurately estimated jobs to gain high priority, since they will be assumed to be delayed by the system rather than their real runtime. 
For our environment, we dynamically adjust the grace period to last 2 resource seconds, as the resource may slightly fluctuate over time. The user will be revived if the inequality $V_{global} < V^{k}_{global,\ end} + T_{grace} * R$ 
is satisfied, where $V^{k}_{global,\ end}$ is the global virtual end time of user $k$, $R$ is the total system resource amount, and $T_{grace}$ is the grace period in resource seconds. The amount of 2 resource seconds has been shown to work for our environment, however, it can differ per system and may need to be more dynamic for appropriate behavior.

\section{Evaluation}

We evaluate the performance of \uwfq on micro-benchmarks and macro-benchmarks. We compare \uwfq with the default Spark built-in fair scheduler and with other representative schedulers. Additionally, we quantify the impact of runtime partitioning on response time and fairness.

\subsection{Setup}

Experiments are conducted on a 
the DAS-5 cluster~\cite{das5_paper}. 
We reserve 5 computing nodes equipped with dual 8-core processors and 60 GB of RAM. Communication between nodes is facilitated by Gigabit Ethernet (GbE) and InfiniBand (IB). We use 4 of the 5 reserved computing nodes to spawn executors, where each node would have at most 2 executors, each reserving 4 cores and 4 GB of RAM, with a total of 32 cores and 32 GB across the cluster. The last node was reserved to run the driver program, with 8 cores and 60 GB of memory. 
We run our experiments on our modified Spark 3.5.5 version
built with Scala 2.12. The modifications applied to Spark enable custom scheduler and partitioner support and do not interfere with the performance of built-in schedulers. The driver is compiled with Java 17 and runs with the Spark standalone cluster manager. 
We leave most Spark configuration parameters at default values with a few exceptions. We enable event logging to collect execution traces after the application has finished, and set high job retention thresholds to avoid omitting earlier jobs. For our partitioned dataset, Spark's default settings cause the dataset to be partitioned too extensively, creating task scheduling overhead.
We solve this by increasing the $maxPartitionBytes$ to an empirically measured amount to avoid this. We have released the experiment setup source code\footnote{Experiment source code: \href{https://github.com/kazemaksOG/spark-benchmark-tool}{https://github.com/kazemaksOG/spark-benchmark-tool}}. 
Designing an accurate runtime predictor is orthogonal to our work. We  therefore assume a perfect runtime prediction for our experiments and discuss this assumption in Sec.~\ref{sec: job runtime prediction}.

\subsubsection{Metrics}
With our experiments, we aim to demonstrate that the \uwfq scheduler can improve response time while still ensuring bounded user-job fairness. To measure this, we collect job response times, slowdowns, and deadline violations and slack ratios. 

\noindent $\bullet$ The \textbf{response time} (RT) of an analytics job is the time that elapses since its first stage is submitted until the last stage is completed, i.e., $RT_i = \Call{max}{T^{i}_{s,end}} - \Call{min}{T^{i}_{s,start}}$.

\noindent $\bullet$ The \textbf{Slowdown} (SL) is calculated by dividing the response time of the job running in the schedule with the response time when the same job is run in the idle system, i.e., as $SL_i = \frac{RT^i_{shared}}{RT^i_{idle}}$.
Compared to response times, slowdowns show relative gains and losses in performance instead of absolute units.

\noindent $\bullet$ The \textbf{deadline violations ratio}~(DVR) is computed by first calculating the ratio $r_i$ for each job by taking the end time difference between the target scheduler and UJF and normalizing it by the UJF runtime of the job, as seen in~\autoref{eq: deadline ratio calc}. Then to calculate the DVR value, the average of the of incurred proportional violations is taken, as seen in~\autoref{eq: deadline violation calc}. Since we do not have a "true" UJF running in the system, we implement a practical UJF scheduler in Spark, and use its execution trace as a substitute. This metric essentially gives us an overview of how much the target scheduler slows down the jobs in comparison to a UJF scheduler.

\noindent $\bullet$ The \textbf{deadline slack ratio}~(DSR) calculates the average of proportional slack time gained, as seen in~\autoref{eq: deadline slack calc}. This showcases the speed-up that the target scheduler brings in comparison to UJF. 

\begin{equation}
    \label{eq: deadline ratio calc}
    r_i = \frac{\Call{max}{T^{i}_{s,end,\text{target}}} - \Call{max}{T^{i}_{s,end,\text{UJF}}}}{RT^i_{\text{UJF}}}
\end{equation}

\begin{equation}
    \label{eq: deadline violation calc}
    DVR = \frac{\sum_i \max(0, r_i)}{\sum_i 1_{\{r_i > 1\}}}
\end{equation}

\begin{equation}
    \label{eq: deadline slack calc}
    DSR = \frac{\sum_i \max(0, -r_i)}{\sum_i 1_{\{r_i \leq 1\}}}
\end{equation}

\subsubsection{Baseline Schedulers}

In micro-benchmarks, we compare \uwfq to representative schedulers. We show the general priority formulas for each of these schedulers. Note that in Spark, a stage with the lowest priority value $P_s$ has the highest scheduling priority.

\noindent $\bullet$ \textbf{Fair scheduling.} Spark comes in with a built-in Fair scheduler that assigns the highest priority to stage $s$ with the least amount of running tasks $N^{s}$, using $P_{s} = N^{s}_{active\ task\ amount}$.
While this is the scheduler we aim to replace, it does not implement UJF, hence it does not necessarily give a good indication of fairness. 

\noindent $\bullet$ \textbf{UJF scheduler.} We implement a practical UJF scheduler in Spark, which we use as the baseline for fairness. This is done by dynamically creating pools for each user as they arrive, assigning the highest priority to the user $k$ with the least amount of tasks $N^k$, with $P_{k} = N^{k}_{active\ task\ amount}$.
Internally, pools use Fair scheduling to distribute resources among their active stages. Note that this does not implement perfect fairness, since we are working with real hardware, where perfect parallelism is not achievable.

\noindent $\bullet$ \textbf{Cluster Fair Queuing.} We implement Cluster Fair Queuing (CFQ)~\cite{chen_scheduling_fair_cluser_spark} for our benchmark by updating the public source code published on GitHub~\footnote{CFQ implementation: \href{https://github.com/chenc10/spark-CFQ-INFOCOM17}{https://github.com/chenc10/spark-CFQ-INFOCOM17}}. While CFQ implements GPS bounded fairness, it is still valuable to compare \uwfq, to see if our proposals bring practical benefit to the system. CFQ priority is decided similarly to \uwfq, but only considers fairness across stages by assigning a deadline $D$ based on traditional virtual time~\cite{parekh_scheduling_original_weighted_fair_queuing}, using $P_{s} = D_{s}$, 
omitting the wider context of users and jobs.

\noindent $\bullet$ \textbf{Using runtime partitioning.} All schedulers are initially evaluated using the default partitioning of Spark. However, to highlight the impact of using runtime partitioning, each scheduler is also evaluated while employing our partitioning algorithm. To distinguish them from their original version, we add a -P to the end of the schedulers utilizing runtime partitioning, e.g., \uwfq-P. Note that when calculating DVR and DSR values, we compare finish times to UJF with the same partitioning implementation.

\begin{table*}[t]
\centering
\caption{Comparison of scheduler performance metrics for scenarios 1 and 2.}
\label{table: micro-benchmark scenario 2}
\resizebox{0.9\textwidth}{!}{%
\begin{tabular}{c|c|cccc|cccc|cc|}
\cline{2-10}
& \textbf{Scheduler} & \multicolumn{4}{c|}{\textbf{Response time (s) $\downarrow$ | \textbf{Slowdown} $\downarrow$}} & \multicolumn{4}{c|}{\textbf{Fairness}} \\
\cline{3-10}
& & \textbf{Avg.} & \textbf{Worst 10\%} & \textbf{Freq. | Infreq.} & \textbf{First | Last} 
& \textbf{DVR $\downarrow$} & \textbf{Violation \# $\downarrow$} & \textbf{DSR $\uparrow$} & \textbf{Slack \# $\uparrow$}
\\
\hline

 \parbox[t]{2mm}{\multirow{3}{*}{\rotatebox[origin=c]{90}{Scen. 1}}} & Fair & 44.0 \textbar{} 19.9 & 86.1 \textbar{} 39.0 & 44.5 \textbar{} 40.8 &  - & 3.25 & 141 & 0.27 & 227 \\
& UJF & 43.1 \textbar{} 19.5 & 74.7 \textbar{} 33.9 & 48.8 \textbar{} 5.13 & -  & - & - & - & - \\
& CFQ & 32.5 \textbar{} 14.7 & \textbf{49.8} \textbar{} \textbf{22.5} & \textbf{32.3} \textbar{} 34.2 & - & 3.69 & 91 & \textbf{0.38} & 277 \\
& UWFQ (this work) & \textbf{29.4} \textbar{} \textbf{13.3} & 50.3 \textbar{} 22.8 & 33.1 \textbar{} \textbf{4.50} & -  & \textbf{0.23} & \textbf{58} & 0.37 & \textbf{310} \\ 

\hline

\parbox[t]{2mm}{\multirow{3}{*}{\rotatebox[origin=c]{90}{Scen. 2}}} & Fair & 28.1 \textbar{} 32.9 & 41.6 \textbar{} 48.9 & - & 21.4 \textbar{} 33.7  & 0.78 & 95 & 0.31 & 145 \\
& UJF & 29.1 \textbar{} 34.1 & \textbf{41.3} \textbar{} \textbf{48.5} &  - & 25.4 \textbar{} 33.3  & - & - & - & - \\
& CFQ & 43.2 \textbar{} 50.7 & 49.5 \textbar{} 58.1 & - & 36.7 \textbar{} 47.3  & 0.80 & 211 & \textbf{0.51} & 29 \\
& UWFQ (this work) & \textbf{25.5} \textbar{} \textbf{29.9} & 46.5 \textbar{} 54.6 & - & \textbf{15.8} \textbar{} \textbf{31.3} & \textbf{0.50} & \textbf{94} & 0.43 & \textbf{146} \\ \hline

\end{tabular}
}

\end{table*}

\subsection{Micro-benchmarks}

Our microbenchmarks use For-Hire Vehicle High Volume (FHVHV) Trip Records from the NYC Taxi and Limousine Commission (TLC)\footnote{TLC trip dataset:~\href{https://www.nyc.gov/site/tlc/about/tlc-trip-record-data.page}{https://www.nyc.gov/site/tlc/about/tlc-trip-record-data.page}} datasets of August 2024. 
The data is stored in a Parquet file format, which we further partition on $PULocationID$ to create more row groups, so the file can be split into more partitions by Spark. The total size of the partitioned parquet file is 752 MB, with 19.1M rows. 
 
We simulate analytics jobs by applying a varying number of operations per row of the dataset. A single analytics job consists of 3 phases: loading the dataset, applying the computation, and retrieving results. Each of these phases has its own stages, creating a linear stage dependency tree. We load the same TLC dataset for all jobs, however, we load them separately, so each job has its own dataset loaded without reusing others. The main time spent in each job is applying computations, which can range from sub-second to 10 seconds in wall-clock time. Finally, the results are collected, which takes only a couple of milliseconds.

We define tiny and short jobs, where each type of job always performs the same operations. We use their runtimes in the idle system
for slowdown calculations. In such settings, short, and tiny jobs respectively require 2.25, and 0.90\,s to run. While this may not accurately simulate dynamic runtimes seen in real settings, it is sufficient to highlight the interaction between jobs based on their arrival times and system congestion. In one of the scenarios, we mimic infrequent user behavior by using a Poisson distribution. A Poisson distribution is commonly used in scheduling to simulate user action frequency~\cite{ilyushkin_scheduling_unknwon_runtime, ilyushkin_scheduling_impact_of_unknown_runtime,pastorelli_scheduling_virtual_sizes}, and it gives fine control over user workload without hard-coding arrival times. 

\begin{table*}[h!]
\centering
\caption{Comparison of scheduler performance metrics with the macro-benchmark.}
\label{table: homo macro}
\resizebox{.85\textwidth}{!}{%
\begin{tabular}{|c|c|cccc|cc|cc|}
\hline
\textbf{Scheduler} & \textbf{Runtime $\downarrow$} & \multicolumn{4}{c|}{\textbf{Response time (s) $\downarrow$}} & \multicolumn{4}{c|}{\textbf{Fairness}} \\
\cline{3-10}
& & \textbf{Avg.\%} & \textbf{0--80\%} & \textbf{80--95\%} & \textbf{95--100\%}
& \textbf{DVR  $\downarrow$} & \textbf{Violation \# $\downarrow$} & \textbf{DSR  $\uparrow$} & \textbf{Slack \# $\downarrow$}
\\
\hline

Fair & 563.1 &46.12 &28.29 &100.5 &\textbf{164.4} &1.00 & 83 & 0.42 & 91 \\
UJF & 565.7 &54.12 &27.43 &140.9 &215.5 &- & - & - & - \\
CFQ & 621.0 &37.05 &11.05 &85.67 &298.1 &0.55 & 30 & 0.52 & 144 \\
UWFQ (this work) & 624.2 &41.42 &12.33 &92.36 &343.5 &\textbf{0.44} & 38 & 0.51 & 136 \\ \hline
Fair-P & 563.8 &49.66 &27.94 &112.7 &202.7 &1.30 & 102 & 0.28 & 72 \\
UJF-P & \textbf{562.8} &50.44 &23.02 &140.2 &214.4 &- & - & - & - \\
CFQ-P & 592.5 &\textbf{28.64} &\textbf{5.51} &\textbf{63.78} &284.2 &0.69 & 28 & \textbf{0.61} & 146 \\
UWFQ-P (this work) & 591.8 &31.19 &6.05 &67.44 &314.6 &0.61 & 27 & 0.56 & 147 \\ \hline

\end{tabular}
}

\end{table*}

\subsubsection{Scenarios}

In collaboration with industry specialists, we consider two scenarios that have been observed in production. 

\noindent $\bullet$ \textbf{Scenario 1: infrequent and frequent users}.
We first examine how schedulers handle the case where some users have many concurrent jobs while some other users only infrequently come to the system to run their workload. This is the main scenario that \uwfq is trying to improve by introducing user context in scheduling, where more emphasis is put on resource distribution among users rather than the runtime of jobs. The main issues with this scenario arise when infrequent users are left without any resources, or all jobs are computed simultaneously, causing a significant increase in response times.
We construct this scenario by introducing 2 infrequent users and 2 frequent users in the system. Infrequent users follow a Poisson distribution when scheduling their jobs. A frequent user schedules a burst of short jobs every 30 seconds that fully congests the system, introducing delay for all jobs. 

\noindent $\bullet$ \textbf{Scenario 2: multiple frequent users}.
We explore the scenario where multiple users supply a burst of jobs to the system, and how effectively the system can recover while providing fair service to all users. The main issues of this scenario are favoring some users more than others, and attempting to compute all jobs simultaneously, which can increase response times.
We implement this scenario by simply having 4 users schedule many tiny jobs simultaneously. Each user has a predefined start delay to ensure that the arrival order of users is consistent across different runs, however, job completion order may differ.

\subsubsection{Results}

\autoref{table: micro-benchmark scenario 2} summarizes the micro-benchmark results.
UWFQ achieves the best average response time in both scenarios, lowering the mean response time by up to 32\% compared to UJF. We attribute these improvements to both job-context and user-context that are present in \uwfq. 
When utilizing job-context, jobs are favored to run till completion, instead of interleaving between all active jobs. This is best visualized in scenario 2 where all jobs are completed gradually, rather than in batches, as seen in~\autoref{fig: micro scenario 4 cdfs}. We can also see that for this scenario, CFQ performs the worst out of all schedulers by a significant amount. We attribute this to the fact that CFQ does not utilize job context, hence executes each job one stage at a time, and finishes them all only at the very end.

\begin{figure*}
\centering
\begin{minipage}{.5\textwidth}
    \centering
    \includesvg[width=1\linewidth]{resources/infrequent_ecdf.svg}
    \caption{Empirical CDFs for infrequent users in scenario 1. }
    \label{fig: micro scenario 2 user cdfs}
\end{minipage}%
\begin{minipage}{.5\textwidth}
  \centering
    \includesvg[width=1\linewidth]{resources/scenario2_ecdf.svg}
    \caption{Empirical CDFs in scenario 2.}
    \label{fig: micro scenario 4 cdfs}
\end{minipage}
\end{figure*}

The most substantial improvement can be seen by introducing user context, where infrequent users have a much better response time in \uwfq and UJF, lowering the average response time in \uwfq by 89\% compared to Fair, which is also highlighted in~\autoref{fig: micro scenario 2 user cdfs}. We believe that in scenarios where users differ in the amount of scheduled jobs, user context allows for fair resource distribution across users, and algorithms that perform fairness only with respect to jobs will disproportionally allocate resources to users with more jobs. This is one of the main drawbacks observed in CFQ, which in this scenario increases the response times of infrequent user jobs by more than 7 times compared to \uwfq.

In scenario 1, 
\uwfq achieves the lowest number of deadline violations and the smallest DVR value.
However, in scenario 2, the fairness metrics are not as representative. This is due to direct completion time comparisons between jobs in the target and UJF schedules, and because the arrival times of all jobs are within a very small interval of time, the order of job completion differs between all schedulers arbitrarily, since there is no global priority between these jobs across schedulers. However, we can analyze the response times between different users to see how resources are distributed among them. 

\uwfq achieves the best performance for its first arriving user and its last arriving user. 
We do note that the first arriving user has a much lower average response time than the last arriving user with \uwfq, however, the same pattern is observed in UJF, hence it could happen due to slightly quicker arrival time of the first user, rather than scheduling unfairness.

\subsection{Macro-benchmarks}

To the best of our knowledge, there is no accepted workload to benchmark batch processing schedulers in multi-user and multi-job environment. Popular benchmark suites such as TPC-H\footnote{TPC-H specification: \href{https://www.tpc.org/tpch/}{https://www.tpc.org/tpch/}}
or HiBench\footnote{HiBench specification: \href{https://github.com/Intel-bigdata/HiBench}{https://github.com/Intel-bigdata/HiBench}}
do not contain the user context information that is necessary to represent such an environment. 
We therefore convert existing system traces that contain multi-user and multi-job metadata into compatible Spark job execution traces, as done previously in some works~\cite{wei_chen_scheduling_spark_preemption,chen_chen_scheduling_effective_partitioning,chen_scheduling_fair_cluser_spark}.
To compose our macro-benchmark, we use Google traces (2014)~\cite{google_traces_wta_format_2014} originally taken from Google cluster usage traces~\cite{wilkes_background_google_trace} but standardized into Workflow Trace Archive (WTA) format~\cite{versluis_background_wta}. 
We select a period of 500 seconds within the trace, and scale the tasks within it to reach a certain utilization threshold. 
We slightly refine the traces to create workloads that we target to solve, while keeping realistic user arrival times and behavior. 
We select the jobs that occur between 1'473'800'000\,ms and 1'474'300'000\,ms. We filter out any jobs whose runtime is 10x longer than the initial median and scale the rest to achieve around 100\% or above theoretical resource utilization. Resource utilization is high to ensure that there is always competition between resources, allowing us to analyze how the system can fairly recover from a burst of jobs. 
The final dataset contains 25 users, where the majority only schedule infrequent small jobs, and 5 large users whose jobs represent more than 90\% of the total workload.  

\subsubsection{Results}

\autoref{table: homo macro}
reports job response times. We group the jobs into the first 80th percentile, the next 15th percentile (medium-sized jobs), and the final 5th percentile.
UJF and Fair execute within the same margin of time, while there is a significant slowdown for both CFQ and \uwfq, both increasing the benchmark time by 10\% in comparison to UJF. We attribute this slowdown to long-running tasks, which may increase the runtime due to ineffective parallelization. However, the slowdown goes down to 5\% when introducing runtime partitioning, which then is most likely caused by JVM warmup~\cite{chen_chen_scheduling_effective_partitioning}, where Fair and UJF algorithms tend to better utilize warm executors.

CFQ and \uwfq achieve the best overall response time, which is attributed to the significant speedups seen for small and medium length jobs.
They respectively improve the average response time by 32\% and 24\% compared to UJF. For jobs in the 80th percentile, the response time is decreased by 60\% and 55\% compared to UJF, and for medium jobs it is lowered by 40\% and 34\% respectively. However, there is a major gain for longer job response times in CFQ and \uwfq, with 38\% and 60\% increase compared to UJF.

\uwfq has a slightly lower DVR value than CFQ, but CFQ has slightly higher DSR value. However, for this benchmark we further analyze the average response times of users to determine fairness differences. We visualize this in~\autoref{fig: macro homo deadline boxplot} by calculating the DVR and DSR values between user average response times rather than job runtimes. We observe that \uwfq is able to more tightly contain the response time improvements with respect to users compared to CFQ. However, the improvement is much less significant than what we have seen in micro-benchmarks.

\begin{figure}[t]
    \centering
    \includesvg[width=1\linewidth]{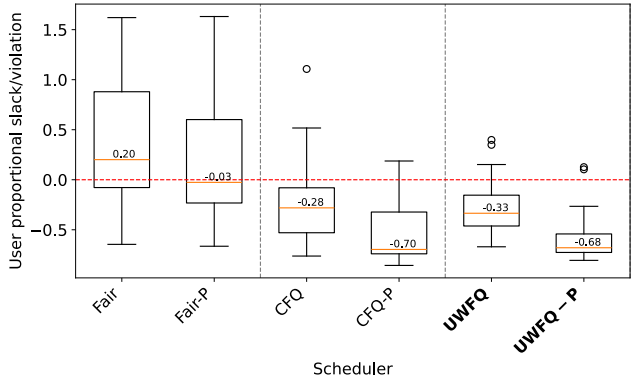}
    \caption{Proportional deadline violations of users, where slacks are negative values and violations are positive. }
    \label{fig: macro homo deadline boxplot}
\end{figure}

Runtime partitioning for a homogeneous workload can significantly improve performance metrics. CFQ-P and \uwfq-P can further reduce average response times by 43\% and 38\% compared to UJF-P, respectively. This massive improvement is caused by lowering the average response times of the 80th percentile by 76\% and 74\% compared to UJF-P, and lowering the next 15th percentile by 55\% and 52\%.
While CFQ and \uwfq still increase the runtime of long jobs, partitioning slightly decreases them to 32\% and 46\% compared to UJF-P. 
Lastly, partitioning additionally makes the algorithms slightly fairer, by lowering the amount of violations and DVR values compared to regular partitioning for both CFQ and \uwfq, and improving fairness across users, as seen in~\autoref{fig: macro homo deadline boxplot}.

\section{Related Work}

\subsection{Scheduling Algorithms}
\label{sec: fair scheduling}

\textbf{Generalized processor sharing (GPS)} is an idealized scheduling model that achieves perfect job-level fairness~\cite{li_scheduling_gps}. GPS 
allocates resources $R_i$ proportionally to a job's weight $w_i$ following $R_i = w_i \times R$.

In practice, GPS cannot be implemented exactly since it requires jobs to be divisible into infinitesimally small quanta, and all resources to run scheduled job quanta simultaneously. However, it is a theoretical baseline that can be used to evaluate other scheduling algorithms.
The major difference between regular GPS and our user-job fairness definition is that job weights are tied to two separate entities. In GPS, once a job enters or leaves the system, the job shares are taken from or distributed among all active jobs on the system. However, under user-job fairness, only the jobs associated with the same user will have to share their resources.

\textbf{Fair-share scheduling} equally distributes resources among users~\cite{kay_scheduling_og_fair}. Jobs are associated with a user $k$ and are serviced proportionally to the user resource share $R_k$ and the job scheduling policy that assigns a weight $w_i$ for the job. In a system with $N_u$ users, the resource $R_i$ each job receives is computed as $R_k = \frac{R}{N_u}$ and $R_i = w_i \times R_k$.

This ensures that each user receives their fair share of resources, guaranteeing progression for their work. While providing a fairness definition for user-level, it leaves job-level scheduling to be defined in the implementation. 

\textbf{Max-min fairness} ensures that each user gains more resources if and only if it does not result in a decrease of resources for another user with a lower or equal share~\cite{tassiulas_other_max_min}. In other words, it maximizes the minimum amount of resources $R_k$ for each user $k$, ensuring fairness, while still distributing resources proportionally by demand. 
This property ensures that each user is guaranteed their fair share of resources, while allowing more demanding users to gain a larger share if permitted. While max-min fairness improves the throughput of the scheduler, it does not account for job priority.

\textbf{Proportional-fair scheduling} can be used to have finer control between fairness and throughput/priority. Proportional-fair scheduling aims to maximize a performance metric of the system while still ensuring at least minimal level of service to all users. 
A common implementation assigns a scheduling interval to the user with the highest priority metric~\cite{sassan_scheduling_proportional_fair}. For the multi-user and multi-job environment, it is defined by $P_k = \frac{P_i}{H_i^\alpha}$,
where $P_k$ is priority of user $k$, $P_i$ is the highest priority job $i$ of the user $k$, and $H_i$ is the historical resource usage time of this user. $\alpha$ is used to tweak the fairness of the algorithm. If set to 0, it will always service the user with the highest job priority. By setting it proportionally high, it gives service to the least serviced user requesting resources, achieving max-min fairness. By setting a reasonable value for $\alpha$, a good compromise between job priority and fairness can be achieved.

\textbf{Weighted fair queuing (WFQ)}~\cite{parekh_scheduling_original_weighted_fair_queuing}
\label{sec: wfq}
executes jobs in the order of their GPS expected finish times and uses job weights to prioritize some jobs.
WFQ first computes the virtual job completion time $V_f$ for every active job in the system as $V_f = V^i_a + \frac{L_i}{w_i}$,
where $w_i$ is the weight of job $i$, $V^i_a$ is the virtual arrival time of job $i$ and $L_i$ is the total execution time of job $i$. Afterward, WFQ schedules the job with the earliest virtual finish time $V_f$.
WFQ approximates GPS, ensuring that every job is completed no later than it would be under GPS, with maximum delay bounded by the size of the largest unsplittable job in the system. This is also expressed by $f_i - \hat{f}_i \leq \frac{L_{max}}{R}$, 
where $f_i$ is the completion time of the job $i$ under WFQ, $\hat{f}_i$ is the completion time of the job $i$ under GPS, $L_{max}$ is the maximum job runtime, and $R$ is the amount of resources that can execute a job in parallel (typically the amount of cores).
WFQ closely approximates the perfect fairness of GPS and regulates the priority of jobs by adjusting their weights. However, its implementation depends on the concept of virtual time.

\textbf{Virtual time}~\cite{chen_scheduling_fair_cluser_spark,parekh_scheduling_original_weighted_fair_queuing,} simulates the marginal service rate that each job receives under GPS. While it is not necessary to enable WFQ, it decreases the time complexity associated with maintaining the finish times of jobs.
Under a GPS scheduling, whenever a new job arrives, each existing job has its finish time extended by the proportion of the resources they have to forfeit to the arriving job. The same happens whenever a job leaves the system, spreading its share across all active jobs. Both events imply an  $O(N)$ complexity to recalculate the finish times of $N$ active jobs.
However, using virtual time, it is possible to warp the time itself based on the amount of resources that are being shared across jobs. Since all jobs equally forfeit or gain resources as jobs enter or leave the system, we can instead advance or slow down virtual time based on the number of active jobs in the system. This is expressed in the time domain by $V(t) = \int_{0}^{t} \frac{R}{N^t_j} \, dt$, 
where $N^t_j$ are the currently active jobs in the GPS schedule at time $t$ and $V(0) =  0$. 
When a job arrives in the system, its virtual deadline is computed as the time at which it would finish under GPS.
While a virtual deadline does not directly map to real time, sorting by virtual deadlines would yield the same order of finish times under GPS scheduling. 
By manipulating virtual time instead of real time, the time complexity incurred by new or ending jobs is reduced from $O(N)$ to $O(\log_2 N)$, since virtual deadlines are set only once, and the incoming jobs have to be put in an ordered list and ending jobs have to be removed from the list, which takes at most $O(\log_2 N)$ time. 

\subsection{Scheduling in Batch Processing Systems}

Pastorelli et al.~\cite{pastorelli_scheduling_virtual_sizes} described a scheduler that is analogous to Weighted Fair Queuing for Hadoop~\cite{parekh_scheduling_original_weighted_fair_queuing}, and that was later adapted to Spark~\cite{chen_scheduling_fair_cluser_spark}. 
These implementations address fairness across jobs rather than among users, as we do in this work.
L. Chen. et al.~\cite{chen_li_scheduling_mix_min_geo_centers} implement a max-min fair scheduling algorithm for Spark jobs across geo-distributed datacenters. 
A fair scheduling solution for application-level scheduling that enables job preemption was developed by W. Chen et al.~\cite{wei_chen_scheduling_spark_preemption}.
Wang et al.~\cite{wang_scheduling_application_level_spark_yarn} improve job completion times by considering deadlines of Spark applications. However, application level scheduling algorithms are not directly applicable to our target industrial system, which is mostly built around a single long-running Spark application.
Some existing implementations attempt to increase the performance of schedulers by optimizing the level of parallelism~\cite{cheng_scheduling_spark_streaming_parallelism,chen_chen_scheduling_effective_partitioning}. In another work, C. Chen et al.~\cite{chen_chen_scheduling_speculative_reservation} use speculative execution and resource reservation to improve scheduling performance. 
These solutions do not enforce a strict fairness policy. However, we believe that these solutions could be integrated with our proposed scheduler to further improve its performance.
While some works mention the performance slowdown caused by long-running tasks~\cite{pastorelli_scheduling_virtual_sizes, chen_scheduling_fair_cluser_spark, wei_chen_scheduling_spark_preemption}, there currently is no effective solution proposed that is able to reduce the impact of long-running tasks while ensuring minimal performance degradation. Our scheduler solves this issue by introducing runtime partitioning.

\subsection{Existing Fairness Definitions}

Most classic scheduling algorithms ensure fairness by distributing resources equally among their users~\cite{kay_scheduling_og_fair, tassiulas_other_max_min, nagle_scheduling_fair_queuing, zhao_scheduling_two_level_fair_lp} or jobs~\cite{li_scheduling_gps, chen_li_scheduling_mix_min_geo_centers, wei_chen_scheduling_spark_preemption}. These fairness definitions compromise other performance metrics such as throughput or response times, hence it can be beneficial to find a compromise between performance and fairness, which is achieved by proportional fairness~\cite{sassan_scheduling_proportional_fair, grandl_scheduling_multi_packing_fairness_knob}. With job deadlines, fairness requires jobs to be serviced before their deadlines~\cite{ilyushkin_scheduling_impact_of_unknown_runtime, chen_scheduling_soft_deadline_latest_start_time_sort, cadorel_scheduling_neardeadline,wang_scheduling_application_level_spark_yarn}.
More recent works define fairness by equalizing some form of detriment caused by sharing resources, e.g., by equalizing the slowdowns each job or user may experience~\cite{jia_scheduling_metaheuristic_fairness, li_scheduling_clustering_fair_metaheuristic}. Bochenina et al.~\cite{bochenina_scheduling_clustering_soft_deadlines} define unfairness as the maximum difference between any two workflow fines, while Rezaeian et al.~\cite{rezaeian_scheduling_fair_multi_workflow} try to equalize the savings each workflow achieves. Ferreira da Silva et al.~\cite{ferreira_scheduling_fairness_nonclairvoyant} define fairness as the difference between the workflow with the most tasks finished and the workflow with the fewest tasks finished.
Fairness is also sometimes defined by some constraint that has to be satisfied. C. Chen et al.~\cite{chen_scheduling_fair_cluser_spark} propose that fairness is a constraint where the difference between the finish time of a task and the finish time of the same task under max-min scheduling never goes beyond a certain constant value, similarly to Pastorelli et al.~\cite{pastorelli_scheduling_virtual_sizes}. In another paper, fairness is used to impose a maximum performance degradation a job can incur by sharing its resources~\cite{chen_chen_scheduling_effective_partitioning}. In the context of backfilling, fairness also refers to the property that lower priority jobs do not delay higher priority jobs~\cite{yuan_scheduling_backfilling_strict_fairness, gomez_scheduling_fattened_backfilling}. 
We find C. Chen et al.~\cite{chen_scheduling_fair_cluser_spark}'s to be the most appropriate fairness definition in batch processing environments. 

\subsection{Job runtime prediction}
\label{sec: job runtime prediction}

Accurate job runtime prediction is crucial for efficient scheduling and partitioning. However, runtime prediction is orthogonal to our focus, and we assumed perfect runtime prediction for our experiments.
Several existing works demonstrate efficient runtime prediction methods that could be integrated with \uwfq. First, by decomposing jobs into smaller units such as operations \cite{wu_performance_prediction_analyitical_sql}, stages \cite{sayeh_performance_prediction_gray_box}, or tasks \cite{gulino_performance_prediction_workflows_apache}, and estimating their runtime based on these individual units. Second, by training a machine learning model using historic data, for example, using regression trees \cite{li_performance_estimation_regression_trees}. Finally, job runtime can be estimated through job simulations \cite{wang_performance_prediction_apache_spark_jobs, wang_performance_prediction_interference, popescu_performance_prediction_predict, venkataraman_performance_prediction_ernest_simulation}. Even with the high accuracy of these methods, it is worth noting that virtual time-based scheduling, which our work builds upon, has been shown to be robust to inaccurate runtime predictions \cite{chen_scheduling_fair_cluser_spark}.

\section{Conclusion}

We presented the User Weighted Fair Queuing (UWFQ) scheduling algorithm, which aims to minimize the average response time of jobs while ensuring that users and their jobs receive a proportional resource share. We additionally introduced a runtime partitioning algorithm to reduce the impact of long-running tasks on fairness and job response times by partitioning the jobs into finer amounts.
We implemented UWFQ with runtime partitioning in Spark to measure its performance and show its viability. Comparing UWFQ to a simple fair scheduler that evenly spreads resources across users, UWFQ with runtime partitioning reduces the average response times of small jobs by up to 74\% in homogeneous workloads, while still ensuring that most jobs are completed within the fairness boundaries. We additionally showed that UWFQ is more robust for handling multi-user, multi-job critical micro-benchmark scenarios, being able to maintain a 6x lower fairness violation ratio compared to other scheduling algorithms.

\printbibliography

\appendix

\section{Correctness of UWFQ}
\label{sec:proof}

In this section, we prove that our UWFQ algorithm correctly implements bounded user-job fairness. We borrow many structural components of the proof from CFQ~\cite{chen_scheduling_fair_cluser_spark} and WFQ~\cite{parekh_scheduling_original_weighted_fair_queuing}, since our algorithm extends on some aspects of these schedulers. Table~\ref{tab:notations} details our notations.

\subsection{Assumptions}

For these proofs, it is necessary to assume an idealized system for user-job fair and 2-level virtual time schedulers, where each job can be divided into infinitesimally small and equal tasks, and all resources can execute them simultaneously and preempt them without scheduling delay. These are equivalent assumptions that are used for hypothetically implementing the GPS scheduler. We define $R$ to represent the available resource amount, which can be infinitely divisible into shares and does not encounter any scheduling delay. 

However, for UWFQ, we assume a limit in job division, which can create a skew in task runtimes, the longest task in the system being $l_{max}$. Additionally, tasks cannot be preempted and have to run till completion once scheduled, creating even more boundaries.

\begin{table}[t]
    \centering
    \caption{Notations}
    \label{tab:notations}
    \begin{tabular}{c| p{6cm}}
    
    \toprule
    \textbf{Notation} & \textbf{Explanation} 
    \\ \midrule
    $_i$ & Job $i$
    \\ \midrule
    $_k$ & User $k$
    \\ \midrule \midrule
    $U$ & User entity
    \\ \midrule
    $N$ & Amount of items
    \\ \midrule
    $S$ & Set of items
    \\ \midrule
    $L$ & Job slot-time (sum of task runtimes)
    \\ \midrule
    $R$ & System resources
    \\ \midrule\midrule
    $f$ & Finish time in 2-level virtual time schedule
    \\ \midrule
    $\hat{f}$ & Finish time in user-job fairness schedule
    \\ \midrule
    $F$ & Finish time in UWFQ schedule
    \\ \midrule\midrule
    $a$ & Arrival time
    \\ \midrule
    $e$ & Full resource access time
    \\ \midrule
    $l$ & Task runtime
    \\
    \bottomrule
    \end{tabular}
    \end{table}

\subsection{Theorems and proofs}

First, we prove that jobs in 2-level virtual (2-LV) time do not finish later than they would in the user-job fair (UJF) scheduler. For this purpose, we first introduce lemmas. 

\begin{lemma}
    \label{lemma: user equivalence}
    Let $U_k$ be an arbitrary user in both 2-level virtual time and user-job fairness. The share $R_k$ this user possesses is equivalent in both schedules at any period of time.
\end{lemma}

\begin{proof}
    Take user $U_k$. Assume there are $N_u$ users that finish before $U_k$ in UJF. Since both 2-LV and UJF distribute shares equally among users and both are work-conserving, users in both schedules progress at the same rate, meaning that all $N_u$ users that finish before $U_k$ in UJF will finish in the same order in 2-LV, and redistribute the same amount of resources to user shares $R_k$ overtime. 
\end{proof}

\begin{lemma}
    \label{lemma: job order}
    Let $U_k$ be an arbitrary user in both 2-level virtual time and user-job fairness. Let $f_i$ be the finish time of job $i$ in 2-level virtual time, and let $\hat{f_i}$ be the finish time of job $i$ in the user-job fair scheduler. For any arbitrary job $i$ in both schedules, user $U_k$ share increase or decrease does not impact the order of completion time $f_i$ and $\hat{f_i}$.
\end{lemma}

\begin{proof}
    Take user $U_k$. Assume user $U_k$ resources $R_k$ are changing due to users joining and leaving the system. Following~\autoref{lemma: user equivalence}, we know that both schedules will have equivalent shares $R_k$ across this period, hence will contribute equal resources to job execution in both schedules. We can express the finish time $f_i$ as

    \begin{equation}
        \label{eq: fj virtual time finish}
        f_i = \frac{\sum_{n=1}^i L_n}{\hat{R_k}} 
    \end{equation}

    where $n$ till $i$ are all jobs that finish before and including $i$, and $\hat{R_k}$ is the volatile resources assigned to user $k$. Additionally, we can derive the finish time for $\hat{f_i}$

    \begin{equation}
        \label{eq: deriving fj hat}
        \begin{gathered}
            \hat{f_1} = \frac{N^k_s}{\hat{R_k}} \cdot L_1\\
            \hat{f_2} = \hat{f_1} + \frac{N^k_s-1}{\hat{R_k}} \cdot (L_2 - L_1)\\
            \hat{f_3} = \hat{f_2} + \frac{N^k_s-2}{\hat{R_k}} \cdot (L_3 - L_2)\\
            ...\\
            \hat{f_i} = \hat{f_{i-1}} + \frac{N^k_s-i+1}{\hat{R_k}} \cdot (L_i-L_{i-1})\\
            \hat{f_i} = \sum^i_{n=1}(\frac{N^k_s-n+1}{\hat{R_k}}\cdot(L_{n} - L_{n-1})) \text{ , where } L_0 = 0
        \end{gathered}
    \end{equation}

    where $N^k_s$ is the number of jobs the user had at the start of the execution. Let's assume $f_i \leq \hat{f_i}$. By substituting it with~\autoref{eq: fj virtual time finish} and~\autoref{eq: deriving fj hat}

    \begin{equation}
        \label{eq: proving R independant}
        \begin{gathered}
            \frac{\sum_{n=1}^i L_n}{\hat{R_k}}  \leq \ \sum^i_{n=1}(\frac{N^k_s-n+1}{\hat{R_k}}\cdot(L_{n} - L_{n-1}))\\
            \textcolor{red}{\frac{1}{\hat{R_k}}}\cdot\sum_{n=1}^i L_n  \leq \textcolor{red}{\frac{1}{\hat{R_k}}}\cdot\sum^i_{n=1}((N^k_s-n+1)\cdot(L_{n} - L_{n-1}))\\
            \sum_{n=1}^i L_n  \leq \sum^i_{n=1}((N^k_s-n+1)\cdot(L_{n} - L_{n-1}))\\
        \end{gathered}
    \end{equation}

    We get that the order of completion $f_i$ and  $\hat{f_i}$ is independent of user resources.
\end{proof}

\begin{theorem}
    Let $f_i$ be the finish time of job $i$ in 2-level virtual time, and let $\hat{f_i}$ be the finish time of job $i$ in the user-job fair scheduler. For every job $i$ in the system, we show that 
    \begin{equation}
        \label{eq: theorem fj smaller fj hat}
        f_i \leq \hat{f_i}
    \end{equation}
\end{theorem}

\begin{proof}
    We consider two cases that cover all jobs in the system. 

    \textit{Case 1:} take arbitrary job $i$ that is already present in the schedule. Since it is present in the schedule, the job is tied to the corresponding user and is ordered within their job set. We show that the job's finish time $f_i$ will not exceed $\hat{f_i}$ under any circumstance.

    First, let's consider that no new job arrives in the system. Following~\autoref{lemma: job order}, we know that user share $\hat{R_k}$ changes over time do not affect the order of finish times. We can then substitute~\autoref{eq: theorem fj smaller fj hat} with~\autoref{eq: fj virtual time finish} and~\autoref{eq: deriving fj hat}

    \begin{equation}
        \small
        \label{eq: proving fj smaller fj hat}
        \begin{gathered}
            \frac{\sum_{n=1}^i L_n}{\hat{R_k}}  \leq \frac{N^k_s}{\hat{R_k}} \cdot L_1 + \frac{N^k_s-1}{\hat{R_k}} \cdot (L_2 - L_1) + ... + \frac{N^k_s-i+1}{\hat{R_k}} \cdot (L_i-L_{i-1})\\
            \textcolor{blue}{\sum_{n=1}^i L_n}  \leq N^k_s \cdot L_1 + (N^k_s-1) \cdot (L_2 - L_1) + ... + (N^k_s-i+1) \cdot (L_i-L_{i-1})\\
            \textcolor{red}{L_1 + L_2 + ... + L_i}  \leq N^k_s \cdot L_1 + (N^k_s-1) \cdot (L_2 - L_1) + ... + (N^k_s-i+1) \cdot (L_i-L_{i-1})\\
            L_2 + ... + L_i  \leq \textcolor{blue}{N^k_s \cdot L_1}  - \textcolor{red}{L_1} + (N^k_s-1) \cdot (L_2 - L_1) + ... + (N^k_s-i+1) \cdot (L_i-L_{i-1})\\
            L_2 + ... + L_i  \leq \textcolor{red}{(N^k_s-1) \cdot L_1} + \textcolor{blue}{(N^k_s-1) \cdot (L_2 - L_1)} + ... + (N^k_s-i+1) \cdot (L_i-L_{i-1})\\
            L_2 + ... + L_i  \leq \textcolor{red}{(N^k_s-1) \cdot L_2} + ... + (N^k_s-i+1) \cdot (L_i-L_{i-1})\\
            ...\\
            L_i  \leq (N^k_s-i+1) \cdot L_i\\
            0  \leq (N^k_s-i) \cdot L_i\\
        \end{gathered}
    \end{equation}

    This inequality holds true for all values $N^k_s \geq i$. Now let's consider the case where a job $c$ arrives for the corresponding user. The job $c$ will either finish before job $i$ or after. If job $c$ finishes after $i$, it is obvious that $f_i \leq \hat{f_i}$, since $f_i$ will not encounter any delay. If job $c$ finishes before $i$:

    \begin{equation}
        \small
        \label{eq: proving fj smaller fj hat inserting job l}
        \begin{gathered}
            L_1 + L_2 + ... + \textcolor{red}{L_c} + ... + L_i \leq \\ N^k_s \cdot L_1 + (N^k_s-1) \cdot (L_2 - L_1) + ... + \textcolor{red}{(N^k_s-c + 1) \cdot (L_c - L_{c-1})} + \\ ... + (N^k_s-i+1) \cdot (L_i-L_{i-1})\\
            ...\\
           L_c + ... + L_i \leq \textcolor{red}{(N^k_s-c + 1) \cdot L_c} + ... + (N^k_s-i+1) \cdot (L_i-L_{i-1})\\
           L_i \leq (N^k_s-i+1) \cdot L_i\\
           0  \leq (N^k_s-i) \cdot L_i\\
        \end{gathered}
    \end{equation}

    \textit{Case 2:} take arbitrary job $i$ that is arriving in the schedule. If the job belongs to a user that does not exist, the job will create a user and assign itself to it. In that case, since both 2-LV and UJF are work-conserving, the job will finish at the same time for both schedules. If the user already exists, the job has to find its order in the existing job set of the corresponding user. Once the job is ordered, \textit{case 1} can be applied to show that $f_i \leq \hat{f_i}$.

\end{proof}

Second, we prove that UWFQ executes jobs in bounded 2-level virtual time.
\begin{theorem}
    Let $f_i$ be the finish time of job $i$ in 2-level virtual time, and let $F_i$ be the finish time of job $i$ in the UWFQ scheduler. For every job $i$ in the system, we show that 
    \begin{equation}
        F_i - f_i \leq \frac{L_{max}}{R} + 2 \cdot l_{max} 
    \end{equation}
\end{theorem}

\begin{proof}
    Let $a_i$ be the arrival time of job $i$, and job $i$ belonging to user job set $S^k_{jobs}$. In UWFQ, since there may be other jobs running at $a_i$, job $i$ tasks may be scheduled later due to priority or task runtime skews. Let $e_i$ represent the time when job $i$ has full access to resources $R$ until completion of all its tasks. Since no other job can interfere with job $i$ during period $e_i$, its runtime is bounded by the longest running task, i.e.
    \begin{equation}
        \label{eq: uwfq worst case response ej}
        F_i \leq e_i + l^i_{max}
    \end{equation}

    By analyzing the period $e_i$ in all cases, we can determine the worst-case response time $F_i$ for job $i$ in UWFQ. We split this into two major cases:

    \textit{Case 1:} The job $i$ immediately receives all resources $R$, hence $a_i = e_i$. In 2-LV, the finish time of job $i$ would be

    \begin{equation}
        \label{eq: 2-lv finish time idle}
        f_i \leq a_i + \frac{L_i}{R_k}
    \end{equation}

    By subtracting~\autoref{eq: 2-lv finish time idle} from~\autoref{eq: uwfq worst case response ej} 

    \begin{equation}
        \label{eq: bounded aj equal ej case}
        \begin{gathered}
        F_i - f_i \leq e_i + l^i_{max} - a_i - \frac{L_i}{R_k}\\
        F_i - f_i \leq a_i + l^i_{max} - a_i - \frac{L_i}{R_k}\\
        F_i - f_i \leq l^i_{max} - \frac{L_i}{R_k} \leq  l_{max} \\
        \end{gathered}
    \end{equation}

    \textit{Case 2:} The job $i$ is delayed to receive all resources $R$, hence $a_i < e_i$. There are two types of jobs that delay $e_i$: jobs that have higher priority than job $i$, and jobs that have lower priority but started before job $i$. Let job index indicate the order of execution in UWFQ, and let set $S_{high}$ be the set of all jobs that execute before $i$ with higher priority, i.e. $S_{high} = \{\ d\ |\ d < i \text{ and } f_d < f_i\ \}$, and let set $S_{low}$ be the set of jobs that execute before job $i$ with lower priority, i.e., $S_{low} = \{\ g\ |\ g < j \text{ and } f_g > f_i\ \}$. 

    \textit{Case 2-1:} First, let's examine the case where there are only higher priority jobs from set $S_{high}$ running before job $i$. In that case, 2-LV would finish 

    \begin{equation}
        \label{eq: 2-LV higher priority jobs}
        f_i \leq a_i + \frac{L_i}{R_k} + \sum_{n \in S^k_{jobs} \cap S_{low}} \frac{L_n}{R_k}
    \end{equation}

    For UWFQ, the finish time of job $i$ also depends on other user jobs that may run before it. We assume a worst-case scenario, where for every user there is a job that finishes just before $f_i$. This means that all other users will be fully utilizing their share of the resource before job $i$ can be scheduled in UWFQ. Since both UWFQ and 2-LV are work-conserving, user $U_k$ will have the same amount of resources time to execute on in both UWFQ and 2-LV, so the job $i$ will be able to receive all resources at the same time it would in 2-LV, i.e.

    \begin{equation}
        \label{eq: ej for user}
        e_i \leq a_i + \sum_{n \in S^k_{jobs} \cap S_{low}} \frac{L_n}{R_k} + C
    \end{equation}

    where $C$ is a delay caused by skews. To account for skews, we can assume a worst-case scenario, where previously running jobs have accumulated a long-running task on all resource instances, creating a worst-case delay of $C = l_{max}$. We can now substitute~\autoref{eq: uwfq worst case response ej} with~\autoref{eq: ej for user} and substract~\autoref{eq: 2-LV higher priority jobs}

    \begin{equation}
        \label{eq: case 2-1 proof}
        \begin{gathered}
            F_i - f_i \leq \textcolor{red}{a_i} + \textcolor{red}{\sum_{n \in S^k_{jobs} \cap S_{low}} \frac{L_n}{R_k}} + C + l^i_{max} -\\ - (\textcolor{red}{a_i} + \frac{L_i}{R_k} + \textcolor{red}{\sum_{n \in S^k_{jobs} \cap S_{low}} \frac{L_n}{R_k}})\\
            F_i - f_i \leq C + l^i_{max} - \frac{L_i}{R_k}\\
            F_i - f_i \leq l_{max} + l^i_{max} - \frac{L_i}{R_k} \leq 2 \cdot l_{max}
        \end{gathered}
    \end{equation}

    \textit{Case 2-2:} Now let's examine the case where the set of both high priority jobs $S_{high}$ and low priority jobs $S_{low}$ runs before job $i$.

    For 2-LV, lower priority jobs cause no delay. The only case a job from $S_{low}$ can run is in the period where all jobs in $S_{high}$ have finished and job $i$ has not arrived yet. As soon as job $i$ arrives, it will preempt the lower priority job, resulting in the same finish time as in~\autoref{eq: 2-lv finish time idle}.

    For UWFQ, a job from $S_{low}$ can also only run in the period where all jobs in $S_{high}$ have finished and job $i$ has not arrived yet. However, since there is no preemption, in the worst case scenario where job $i$ arrives just before a job from $S_{low}$ is scheduled, it will be delayed by the entire runtime of the low priority job, i.e.

    \begin{equation}
        \label{eq: ej delayed by low prior}
        e_i \leq a_i + \frac{L_{max}}{R} + l_{max}
    \end{equation}

    Substituting~\autoref{eq: uwfq worst case response ej} with~\autoref{eq: ej delayed by low prior} and substracting~\autoref{eq: 2-lv finish time idle} 

    \begin{equation}
        \label{eq: case 2-2 proof}
        \begin{gathered}
            F_i - f_i \leq a_i + \frac{L_{max}}{R} + l_{max} + l^i_{max} - (a_i + \frac{L_i}{R_k}) \\
            F_i - f_i \leq  \frac{L_{max}}{R} + l_{max} + l^i_{max} -  \frac{L_i}{R_k} \leq \frac{L_{max}}{R} + 2 \cdot l_{max} \\
        \end{gathered}
    \end{equation}

\end{proof}

Lastly, we establish a corollary that UWFQ is bounded by user-job fairness.

\begin{corollary}
    Since jobs in 2-level virtual time are bounded by user-job fairness, and UWFQ is bounded by 2-level virtual time, we can express that UWFQ is bounded by user-job fairness.
\end{corollary}

\end{document}